\documentclass[10pt]{iopart}

\usepackage{iopams}
\usepackage{citesort} 
\usepackage{psfrag}
\usepackage{amssymb}
\usepackage{amscd}
\usepackage{eucal}
\usepackage{epsfig}
\usepackage{verbatim}
\usepackage{color}
\usepackage{url}
\usepackage{subfigure}

\newtheorem{theorem}{Theorem}

\newtheorem{corollary}[theorem]{Corollary}
\newtheorem{conjecture}[theorem]{Conjecture}
\newtheorem{proposition}[theorem]{Proposition}

\newtheorem{remark}{Remark}

\newtheorem{example}{Example}

\newcommand{\perm}{{\rm Perm}}
\newcommand{\diag}{{\rm diag}}
\newcommand{\Mcap}{{\rm cap}}
\newcommand{\E}[1]{\langle#1\rangle}
\newcommand{\norm}[1]{\parallel #1 \parallel}
\newcommand{\edij}{i \rightarrow j}
\newcommand{\edji}{j \rightarrow i}
\newcommand{\edkl}{k \rightarrow l}

\newcommand{\PMatch}{PM}

\newcommand{\eq}{}
\newcommand{\eqs}{}

\newenvironment{proof}[1][Proof:]{\begin{trivlist}
\item[\hskip \labelsep {\bfseries #1}]}{\end{trivlist}}

\begin{document}

\title{Belief propagation and loop calculus for the permanent of a non-negative matrix}
\author{Yusuke Watanabe$^1$ and Michael Chertkov$^2$}
\address{$^1$ Institute of Statistical
Mathematics, 10-3 Midori-cho, Tachikawa, Tokyo 190-8562 Japan.}
\address{$^2$ Center for Nonlinear
Studies and Theoretical Division, LANL, NM, 87545\\
also New Mexico Consortium, Los Alamos, NM
87544.}
\eads{\mailto{watay@ism.ac.jp}, \mailto{chertkov@lanl.gov}}

\begin{abstract}
We consider computation of permanent of a positive $(N\times N)$ non-negative matrix,
$P=(P_i^j|i,j=1,\cdots,N)$, or equivalently the problem
of weighted counting of the perfect matchings over the
complete bipartite graph $K_{N,N}$.
The problem is known to be of likely exponential complexity.
Stated as the partition function $Z$ of a graphical model, the problem allows exact Loop Calculus
representation [Chertkov, Chernyak '06] in terms of an interior minimum of the
Bethe Free Energy functional over non-integer doubly stochastic matrix of marginal beliefs, $\beta=(\beta_i^j|i,j=1,\cdots,N)$,
also correspondent to a fixed point of the iterative message-passing algorithm of the Belief Propagation (BP) type.
Our main result is an explicit expression of the exact partition function (permanent) in terms of the matrix of BP marginals,  $\beta$, as $Z=\mbox{Perm}(P)=Z_{BP} \mbox{Perm}(\beta_i^j(1-\beta_i^j))/\prod_{i,j}(1-\beta_i^j)$,
where $Z_{BP}$ is the BP expression for the permanent stated explicitly in terms of $\beta$.  
We give two derivations of the formula, a direct one based on the Bethe Free Energy and
an alternative one combining the Ihara graph-$\zeta$ function and the Loop Calculus approaches.
Assuming that the matrix $\beta$ of the Belief Propagation marginals is calculated, we provide two lower bounds and one upper-bound to estimate the multiplicative term. Two complementary lower bounds are based on the Gurvits-van der Waerden theorem and on a relation between the modified permanent and determinant respectively.
\end{abstract}

\submitto{\JPA}


\section{Introduction}

The problem of calculating the permanent of a non-negative matrix arises in many contexts in
statistics,  data analysis and physics. For example, it is intrinsic to the parameter learning of a flow
used to follow particles in turbulence and to cross-correlate two subsequent images
\cite{10CKKVZ}. However, the problem is $\# P$-hard  \cite{79Val},  meaning that solving it
in a time polynomial in the system size, $N$, is unlikely.
Therefore, when size of the matrix is sufficiently large, one naturally looks for ways to approximate the
permanent. A very significant breakthrough  was achieved with invention of a so-called
Fully-Polynomial-Randomized Algorithmic Schemes (FPRAS) for the permanent problem \cite{04JSV}:
the permanent is approximated in a polynomial time, with
high probability and within an arbitrarily small relative error.
However, the complexity of this FPRAS is $O(N^{11})$, making it impractical for the majority of realistic applications.
This motivates the task of finding a lighter deterministic or probabilistic algorithm capable of
evaluating the permanent more efficiently.

This paper continues the thread of \cite{08CKV,10CKKVZ} and \cite{09HJ}, where the Belief Propagation (BP) algorithm was suggested as an efficient heuristic of good (but not absolute) quality to approximate the permanent. The BP family of algorithms,
originally introduced in the context of error-correction codes \cite{63Gal} and artificial
intelligence \cite{88Pea}, can generally be stated for any graphical model \cite{05YFW}.
The exactness of the BP on any graph without loops  suggests that the algorithm can be an
efficient heuristic for evaluating the partition function or for finding a Maximum
Likelihood (ML) solution for the Graphical Model (GM) defined on sparse graphs.  However, in the
general loopy cases one would normally not expect BP to work well,  thus making the heuristic results of \cite{08CKV,10CKKVZ,09HJ} somehow surprising,
even though not completely unexpected in view of existence of polynomially efficient algorithms for the ML version of the problem \cite{55Kuh,92Ber},
also realized in \cite{08BSS} via an iterative BP algorithm. This raises the questions of understanding the performance of BP: what
it does well and what it misses? It also motivates the challenge of improving the BP
heuristics.

An approach potentially capable of handling the question and the challenge was recently suggested in the general framework of GM.
The Loop Series/Calculus (LS) of \cite{06CCa,06CCb} expresses
the ratio between the Partition Function (PF) of a binary GM and its BP estimate in terms of a finite series, in which each term is associated with the so-called generalized loop (a subgraph with all vertices of degree larger than one) of the graph. Each term in the series,  as well as the BP estimate of the partition function, is expressed in terms of a doubly stochastic matrix of marginal probabilities, $\beta=(\beta_i^j|i,j=1,\cdots,N)$, for matching pairs to contribute a perfect matching. This matrix $\beta$ describes a minimum of the so-called Bethe free energy,  and it can also be understood as a fixed point of an iterative BP algorithm.
The first term in the resulting LS is equal to one.
Accounting for all the loop-corrections, one recovers the
exact expression for the PF. In other words,  the LS holds the key to understanding the gap between
the approximate BP estimate for the PF and the exact result.
In section \ref{sec:one} and section \ref{sec:LC},
we will give a technical introduction to the variational Bethe Free Energy (BFE) formulation of BP and a brief
overview of the LS approach for the permanent problem respectively.

{\bf Our results.}
In this paper, we develop an LS-based approach to describe the quality of the BP
approximation for the permanent of a non-negative matrix.
(i) Our natural starting point is the analysis of the BP solution itself conducted in section
\ref{sec:BP}.  Evaluating the permanent of the non-negative matrix, $P=((p_i^j)^{1/T}|i,j=1,\cdots,N)$,
    dependent on the temperature parameter, $ T \in [0,\infty]$, we find
    that a non-integer BP solution is observed only at $T>T_c$, where $T_c$ is defined by
    \eq(\ref{CritEq}).
(ii) At $T>T_c$, we derive an alternative representation for the LS in section \ref{sec:Per_BP_Per}.
    The entire LS is collapsed to a product of two terms:
    the first term is an easy-to-calculate function of $\beta$, and the second term is the permanent of the matrix, $\beta.*(1-\beta)=(\beta_i^j(1-\beta_i^j))$.
    (The binary operator $.*$ denotes the element-wise multiplication of matrices.)
    This is our main result stated in theorem \ref{LS_new}, and the majority of the consecutive statements of our paper follows from it. We also present yet another, alternative, derivation of the theorem \ref{LS_new} using the multivariate Ihara-Bass formula for the graph zeta-function in subsection \ref{subsec:Ihara}.
(iii) Section \ref{sec:low} presents two easy-to-calculate lower bounds for the LS.
    The lower bound stated in the corollary \ref{Gurvits_bound} is based on the Gurvits-van der Waerden theorem applied to $\mbox{Perm}(\beta.*(1-\beta))$.
    Interestingly enough this lower bound is invariant with respect to the BP transformation,
    i.e. it is exactly equivalent to the lower bound derived via
    application of the van der Waerden-Gurvits theorem to the original permanent. Another
    lower bound is stated in theorem \ref{second_low}. Note,  that as follows from an
    example discussed in the text, the two lower bounds are complementary: the latter is stronger
    at sufficiently small temperatures,  while the former
    dominates the large $T$ region.
(iv) Section \ref{sec:up} discusses an upper bound on the transformed permanent based on the application of the
    Godzil-Gutman formula and the Hadamard inequality. Possible future extensions of the approach are discussed in section \ref{sec:path}.

\section{Background (I): Graphical Models, Bethe Free energy and Belief
Propagation.}\label{sec:one}

Permanent of a non-negative matrix,
$P=((p_i^j)^{1/T}|i,j=1,\cdots,N) \quad  (0\leq p_i^j,\ 0\leq T\leq\infty)$,
is a sum over the set of permutations on $\{1,\ldots,N\}$,
which can be parameterized via binary-component vectors, $\sigma$, corresponding to perfect matchings (PM) on
the complete bipartite graph $K_{N,N}$:
\begin{equation}
\left\{ \sigma=(\sigma_i^j) \in \{0,1\}^{N \times N} \Big|
\forall i:\ \sum_{j=1}^N \sigma_i^j=1, \quad  \forall j:\ \sum_{i=1}^N\sigma_i^j=1
 \right\}.
\end{equation}
This binary interpretation allows us to represent
the permanent as the partition function (PF), $Z$, of a probabilistic model over the set of perfect matchings.
Each perfect matching, $\sigma$, is realized with the
probability
\begin{eqnarray}
\fl
\quad {\cal P}(\sigma)=\frac{1}{Z}P^{\sigma}; \quad
P^{\sigma} \equiv \prod_{(i,j) \in E} (p_i^j)^{\sigma_i^j/T},\
Z\equiv \sum_{\sigma : \PMatch}(p_i^j)^{\sigma_i^j/T} = \perm (P), \label{GM}
\end{eqnarray}
where $E=\{ (i,j) |\ i,j=1,\ldots,N \}$ is the edges of $K_{N,N}$.
In the zero-temperature limit, $T\to 0$,  \eq(\ref{GM})
selects one special ML solution, $\sigma_*=\arg\max_{\sigma} P^{\sigma}$.
(Here and below we assume that $P$ is non-degenerate, in the sense that at $T\to 0$, ${\cal P}(\sigma)\to 0$ for $\forall\ \sigma\neq\sigma_*$.)

For a generic GM, assigning (un-normalized) weight $P^{\sigma}$ to a state $\sigma$, one defines exact variational (called Gibbs, in statistical physics, and Kullback-Leibler in statistics) functional
\begin{eqnarray}
{\cal F}\{b(\sigma)\}\equiv T
\sum_{ \sigma }b(\sigma)\ln\frac{b(\sigma)}{P^{\sigma}}. \label{Gibbs}
\end{eqnarray}
One finds that under condition that the belief, $b(\sigma)$, understood as a proxy to the
probability ${\cal P}(\sigma)$, is normalized to unity, $\sum_{\sigma \in PM} b(\sigma)=1$, the Gibbs
functional is convex and it achieves its only minimum at $b(\sigma)={\cal P}(\sigma)$ and
${\cal F}\{ \mathcal{P} \}=-T\ln Z$.

BP method offers an approximation which is exact when the underlying GM is a tree.
As shown in \cite{05YFW}, the BP approach can also be stated for a general GM as a relaxation of the Gibbs functional (\ref{Gibbs}).
In this paragraph we briefly review the concept of \cite{05YFW} with application to the permanent problem.
For the GM (\ref{GM}), the BP approximation for the state beliefs becomes
\begin{eqnarray}
 b(\sigma)\approx b_{\it BP}(\sigma)=
 \frac{\prod_i b_i(\sigma_i)\prod_j b^j(\sigma^j)}{
 \prod_{(i,j) \in E} b_i^j(\sigma_i^j)},
\label{BP_Belief}
\end{eqnarray}
where $\forall i,j$: $\sigma_i=(\sigma_i^j \in \{0, 1\}|j=1,\cdots, N)$
s.t. $\sum_j\sigma_i^j=1$ and
$\sigma^j=(\sigma_i^j \in \{0, 1\}|i=1,\cdots,N)$
s.t. $\sum_i\sigma_i^j=1$, i.e. $\sigma_i$ and
$\sigma^j$ each has only $N$ allowed states corresponding to allowed local perfect matchings for
the vertices $i$ and $j$ respectively.
The vertex and edge beliefs are related to each other according to
\begin{equation}
\forall (i,j) \in E :\quad b_i^j(\sigma_i^j)= \sum\limits_{\sigma_i\setminus\sigma_i^j}b_i(\sigma_i)=
\sum\limits_{\sigma^j\setminus\sigma_i^j}b^j(\sigma^j),\label{rel}
\end{equation}
and the beliefs, as probabilities, should also satisfy the normalization conditions:
\begin{equation}
\forall (i,j) \in E :\quad b_i^j(1)+b_i^j(0)=1.\label{norm}
\end{equation}
Note,  that our notations for beliefs are not identical to ones used in \cite{05YFW}:
the multi-variable beliefs, $b_i$, are associated with vertexes of $K_{N,N}$, and the single-variable beliefs,
$b_i^j$ are associated with edges of the graph.
Substituting \eq(\ref{BP_Belief}) into \eq(\ref{Gibbs}) and
approximating $\sum_{ \sigma \in PM} b(\sigma) f(\sigma_i^j)$ with $\sum_{\sigma_i^j} b_i^j(\sigma_i^j) f(\sigma_i^j)$ etc,
one arrives at the BFE functional
\begin{eqnarray}
\fl
  {\cal F}_{\it BP}\{b_i^j(\sigma_i^j);b_i(\sigma_i);b^j(\sigma^j)\}
   \equiv E-T S, \quad
  E\equiv\sum_{(i,j)}b_i^j(1)\log(p_i^j),\label{FE}\\
  \fl S\equiv\sum_{(i,j)}\sum\limits_{\sigma_i^j}
  b_i^j(\sigma_i^j)\ln b_i^j(\sigma_i^j)
   -\sum_i\sum\limits_{\sigma_i} b_i(\sigma_i)
  \ln b_i(\sigma_i)
   - \sum_j\sum\limits_{\sigma^j} b^j(\sigma^j)\ln b^j(\sigma^j).\label{S}
\end{eqnarray}
Note that the BFE functional is bounded from below and generally non-convex,
and thus finding the absolute minimum of the BFE
is the main task of the BFE approximation. The BP approximation $Z_{BP}$ of the
partition function is given by ${\cal F}_{BP}=-T\ln Z_{BP}$ at a minimum of the BFE.

Moreover, the variational formulation of \eqs(\ref{rel},\ref{norm},\ref{FE},\ref{S}) can be significantly simplified in our case;
one can utilize \eqs(\ref{rel},\ref{norm}) and express $b_i(\sigma_i), b^j(\sigma^j)$ and $b_i^j(\sigma_i^j)$ solely in terms of the $\beta_i^j\equiv
b_i^j(1)$ variables, satisfying doubly-stochastic constraints
\begin{eqnarray}
 \forall (i,j) \in E : 0\leq \beta_i^j\leq 1; \quad
\forall i:  \sum_j \beta_i^j=1; \quad \forall j: \sum_i \beta_i^j=1. \label{ds_cond}
\end{eqnarray}
The entropy \eq(\ref{S}) becomes
\begin{eqnarray}
\fl
 S\{\beta_i^j\}&=
 \sum_{(i,j)} \left( \beta_i^j \log \beta_i^j +  (1-\beta_i^j) \log (1-\beta_i^j) \right)
-
 \sum_i \sum_j \beta_i^j \log \beta_i^j
-
 \sum_j \sum_i \beta_i^j \log \beta_i^j  \nonumber \\
\fl
 &= \sum_{(i,j)}\left((1-\beta_i^j)\ln(1-\beta_i^j)-\beta_i^j\ln\beta_i^j
      \right). \label{ES}
\end{eqnarray}
%
%
Therefore, the Bethe-Free energy approach applied to the GM (\ref{GM}) results in minimization of the following Bethe-Free Energy (BFE) functional
\begin{equation}
  \mathcal{F}_{BP}\{ \beta \} = T \sum_{(i,j) \in E}
\left(
\beta_i^j\ln\frac{ \beta_i^j }{(p_i^j)^{1/T}}
- (1-\beta_i^j)\ln(1-\beta_i^j)
\right), \label{BFE}
\end{equation}
over $\beta=(\beta_i^j)$ under the constraints \eq(\ref{ds_cond}).

To analyze the minima of the BFE,
we incorporate Lagrange multipliers $\mu_i, \mu^j$ enforcing the constraints in \eqs(\ref{ds_cond}).
Looking for a stationary point of the Lagrange function over the $\beta$ variables, one arrives at the
following set of quadratic equations for each (of $N^2$) variables, $\beta_i^j$
\begin{equation}
 \forall (i,j) \in E :\quad
 \beta_i^j(1-\beta_i^j)=(p_i^j)^{1/T}\exp\left(\mu_i+\mu^j\right). \label{BP1}
\end{equation}
One observes that any solution of \eqs(\ref{ds_cond},\ref{BP1}) at $T>0$, that contains at
least one $\beta_i^j$ which is not integer, does not contain any integers among all $\beta_i^j$.
In fact,  our main focus will be on these non-integer (interior) solutions of
\eqs(\ref{ds_cond},\ref{BP1}).
To find a solution of BP \eqs(\ref{ds_cond},\ref{BP1})
one relies on an iterative procedure. For a description of a set of
iterative BP algorithms convergent to a minimum  of the BFE for the perfect matching problem we
refer the interested reader to \cite{08CKV,10CKKVZ,09HJ}.

\begin{remark} 
{\rm
Note that just derived BP approximation differs from the so-called Mean-Field (MF) approximation corresponding to the following ansatz
\begin{eqnarray}
 b(\sigma)\approx b_{\it MF}(\sigma)=
 \prod_{(i,j) \in E} b_i^j(\sigma_i^j),
\label{BP_MF}
\end{eqnarray}
enforcing statistical independence of the edge beliefs.
If one substitutes $b(\sigma)$ by $b_{\it MF}(\sigma)$ in \eq(\ref{Gibbs}) and also accounts for the normalization condition (\ref{norm}),  which may be understood here as one enforcing the ``Fermi exclusion principle'' for an edge $(i,j)$ to contribute a perfect matching, $\sigma_i^j=1$, the resulting expression for the MF free energy will turn into BP expression (\ref{BFE}) with the first term there changing sign to $-$.
One expects that BP approximation outperforms MF approximation in accuracy.
Consider,  for example, $N=10$ and $\beta_i^j= 1/N$, then
the exact, BP and MF entropies are $\ln(10!) \approx 15.10$, $ 100(.9\ln(.9)-.1\ln(.1)) \approx 13.54$
and $100(-.9\ln(.9)-.1 \ln(.1)) \approx 32.50$, respectively. An intuitive explanation for
MF overestimating the entropy term is related to the fact that MF ignores correlations related to competitions between neighboring edges for contributing a perfect matching.
}
\end{remark}

\section{Threshold Behavior of BP at Low Temperatures} \label{sec:BP}

As discovered in \cite{08BSS}, at $T=0$, properly scheduled iterative version of BP converges
efficiently to the ML solution of the problem. In this context it is natural to
ask the question of how a non-integer solution of BP emerges with a temperature increase. To address this
question, we first consider the following homogeneous example.

\begin{figure}
\centering \caption{This figure contains a set of illustrations based on the homogeneous example \ref{example1}
discussed in the text. $N=10$ and $W=2$ are chosen for these illustrations. Figure~\ref{fig}b
shows $T\ln Z$ for the homogeneous model (red) and respective BP expression, $T\ln Z_{BP}$ (blue) as functions of the temperature, $T$. Green dash line mark $T_c$. figure~\ref{fig}c shows comparison of different estimations of
$\ln(\mbox{Perm}(\beta.*(1-\beta))/\prod_{(i,j)}(1-\beta_i^j))$ vs the temperature parameter
$T$, where $\beta$ is the matrix of marginal beliefs evaluated at a fixed point of BP equations.
Red, Blue, Purple, Green and Dashed-Gray lines show the exact expression, the lower bound of the
corollary \ref{Gurvits_bound}, the lower bound of the theorem \ref{second_low}, the upper bound of the
proposition \ref{up} and the BP expression, respectively. \label{fig}} \subfigure[${\cal F}_{BP}$ vs
$\epsilon$.]{\includegraphics[width=0.35\textwidth,height=0.23\textwidth]{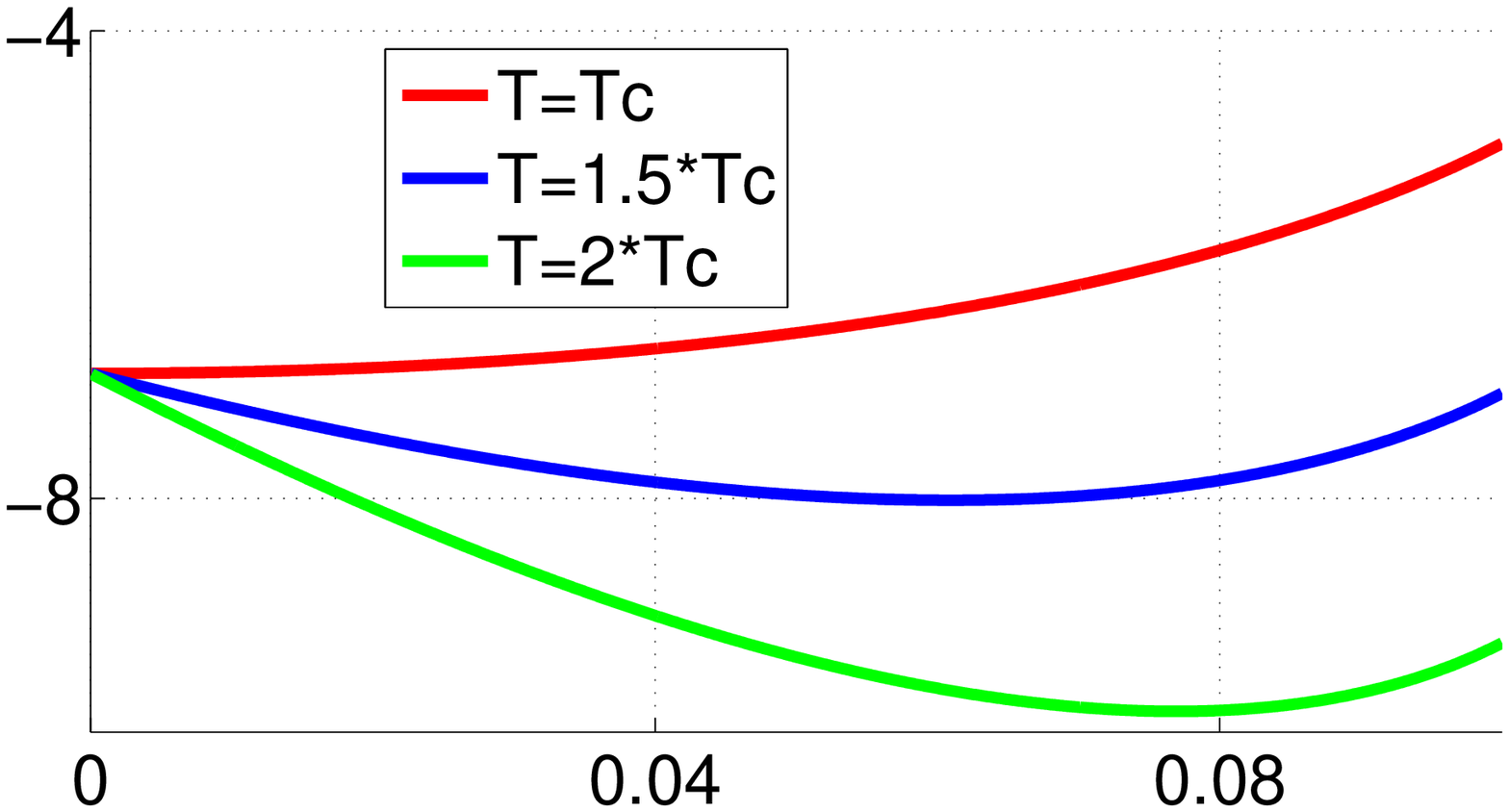}} \subfigure[$T\ln Z$ vs
$T$.]{\includegraphics[width=0.3\textwidth]{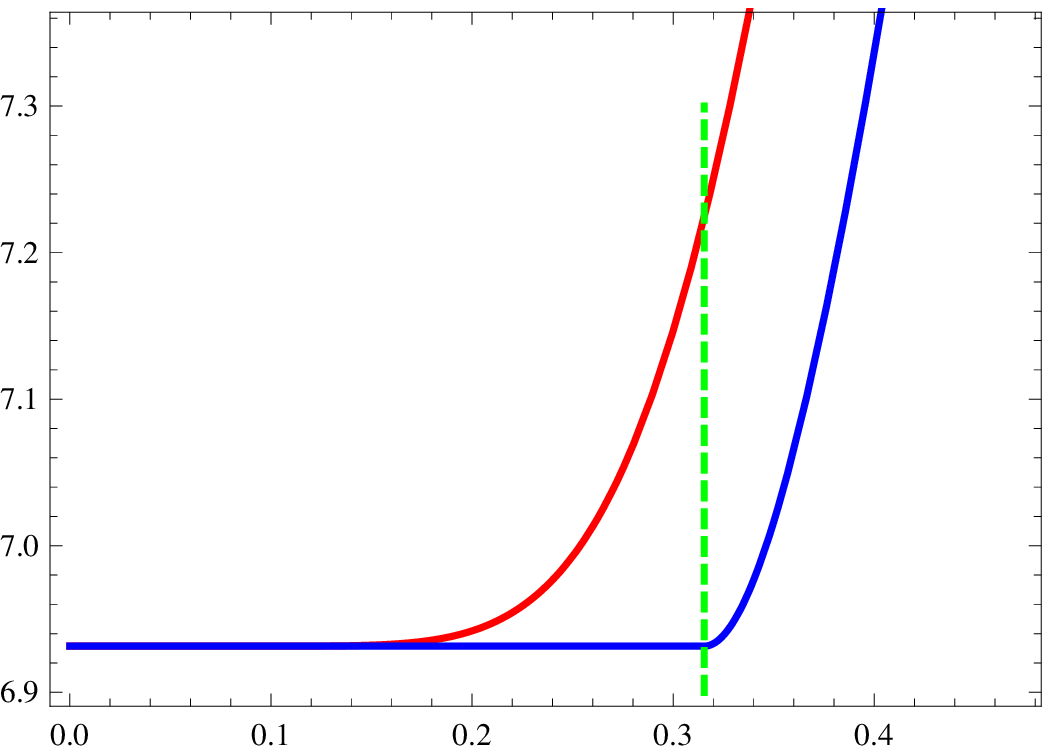}} \subfigure[$\ln(Z/Z_{BP})$ vs $T$ for
different estimators.]{\includegraphics[width=0.3\textwidth]{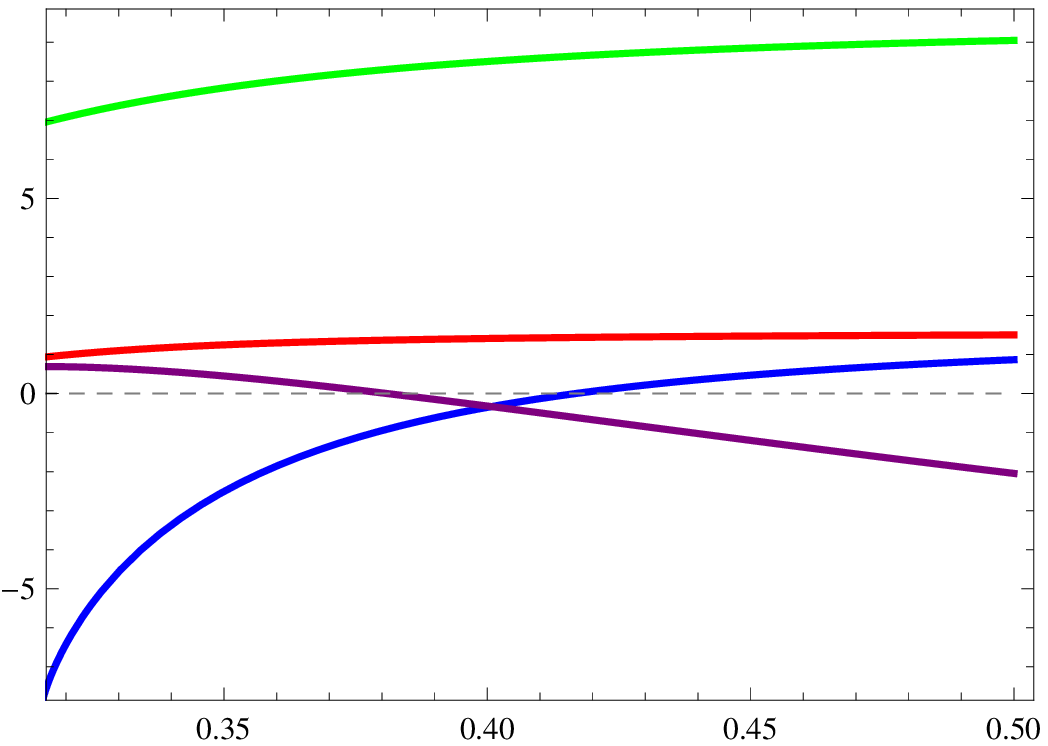}}

\end{figure}

\begin{example}{\rm
\label{example1} Define a homogeneous weight model biased toward a perfect matching solution,
$\sigma_i^j=\delta_i^j$ : $p_i^j=1$ if $i\neq j$ and
$p_i^i=W \ (W>1)$. Looking for $\beta$ in the homogeneous form
\begin{equation}
 \beta_i^j(T)=
\left\{\begin{array}{cc}
1- \epsilon (N-1)    &\mbox{ :if } i=j     \\ \epsilon             &\mbox{   :otherwise, }
\end{array}\right.
\label{example}
\end{equation}
one observes that this ansatz for $\beta$ solves the BP \eqs(\ref{ds_cond},\ref{BP1}) at $\epsilon$ equal to
$\epsilon_{min}= (N-1-W^{1/T})/((N-1)^2
-W^{1/T})$. At $T=\infty$, the probabilities are uniform, i.e. $\beta$ from \eq(\ref{example})
with $\epsilon=\epsilon_{min}$ is $\beta_i^j=1/N$ for all $(i,j) \in E$. Now consider lowering
the temperature and observe that at $T_c=\ln W/\ln (N-1)$ the nontrivial solution, with
$\beta_i^j\neq 0,1$ for all $(i,j) \in E$, turns exactly into the isolated/trivial ML one,
$\beta_i^j=\delta_i^j$. Obviously one finds that the BFE, ${\cal F}_{BF}$, considered as a
function of $\epsilon$, achieves its minimum at $\epsilon=\epsilon_{min}$ if $T>T_c$. Exactly at
$T=T_c$ this $\epsilon_{min}=0$ and the nontrivial solution merges into the isolated ML solution.
The dependence of the BFE on $\epsilon$ for different $T$ (at some exemplary values of $N$ and $W$) is shown in
figure~\ref{fig}a. 
The partition function can be calculated efficiently. Counting the configurations straightforwardly (in a brute force combinatorial manner), one derives
$Z=\sum_{k=0}^N W^{(N-k)/T} {N \choose k} D_k$.
The following recursion is used to evaluate the number of permutations 
coefficient, $D_k$: $\forall\ k \geq 2,\ \ D_k=(k-1)(D_{k-1}+D_{k-2}),\quad D_0=1,\quad D_1=0$. 
A comparison of $T\ln Z$ and $T\ln Z_{BP}$ as functions of $T$ is shown in figure~\ref{fig}b.
}
\end{example}

Returning to the case of an arbitrary nonnegative $P$, we discover that this phenomenon of the
nontrivial solution splitting at some finite nonzero (!!) temperature from the ML configuration
is generic.
\begin{proposition}
\label{prop:Tc} For any non-negative matrix $P=((p_i^j)^{1/T}|i,j=1,\cdots,N)$ one finds a
special (we call it critical) temperature, $T_c$, such that for $T>T_c+\varepsilon$ a nontrivial
solution of BP, corresponding to a local non-saturated minimum of ${\cal F}_{BP}$, dominating
the respective value corresponding to the maximum likelihood solution, is realized for at least
a sufficiently small positive $\varepsilon$.  This special solution coincides with the best
perfect matching solution at $T=T_c$ and it does not exist for $T<T_c$. The critical temperature
$T_c$ solves
\begin{equation}
 \det(P_i^j - 2 \sigma_{* i}^{\ j} P_i^{j})=0 ,  \label{CritEq}
\end{equation}
where $\sigma_*$ is the ML configuration.
\end{proposition}
\begin{proof}
Our proof of the proposition is constructive. Let us look for a solution of the BP equations weakly
deviating from the ML configuration $\sigma_*$.
Without loss of generality we assume that $\sigma_{* i}^{\ j} = \delta_i^j$.
We introduce $v_i^j=\beta_i^j(1-\beta_i^j)\ll 1$ and
observe that a nontrivial solution, approaching the ML one at $v\to 0$, is
$\beta_i^j=(1-(1-2\delta_i^j)[{1-4v_i^j}]^{1/2})/2$. Linearizing the normalization condition,
over $v$ one derives,
$\forall i:   v_i^i= \sum_{j \neq i}v_i^j ;  \quad
\forall j:   v_j^j= \sum_{i \neq j}v_i^j$
On the other hand,  the BP equation (\ref{BP1}), complemented by the set of linear constraints
on $v$,  translates into,
$\forall i: \  P_i^i U^i= \sum_{j \neq i} P_i^j U^j ;  \quad
\forall j: \  P_j^jU_j= \sum_{i \neq j}P_i^j U_i$,
where $U_i=\exp(\mu_i)$ and $U^j=\exp(\mu^j)$. Requiring that the later equations have a nontrivial
solution (with nonzero $v$), one arrives at the critical temperature condition, \eq(\ref{CritEq}).
It is then straightforward to verify that the extension of the nontrivial solution into the $T<T_c$
domain is unphysical (as some elements of the respective small $v$ solution are negative),
while the BFE associated with the nontrivial solution for $T>T_c$ is smaller than the one
corresponding to the ML perfect matching.
\end{proof}

\begin{conjecture}
We conjecture that the non-integer solution of BP equations discussed in proposition \ref{prop:Tc}
extends beyond the small $T_c+\varepsilon$ vicinity of $T_c$, and this solution
transitions smoothly at $T\to\infty$ into the obvious fully homogeneous solution,
$\beta_i^j=1/N$ for all $(i,j)\in E$.  Another plausible conjecture is
that no other non-integer solutions exist at $T<T_c$; therefore when the non-integer solution
discussed in the proposition emerges at $T=T_c$ it, in fact, gives a global minimum of the BFE.
\end{conjecture}

\section{Background (II): Loop Calculus and Series}\label{sec:LC}
Here we consider $T>T_c$ where, according to the main result of the previous section, there exists
a solution of \eqs(\ref{ds_cond},\ref{BP1}) lying in the interior of the doubly-stochastic-matrix polytope.
We assume that such a nontrivial solution of the BP equations is found.

As shown in \cite{06CCa,06CCb}, the exact partition function of a generic GM can be expressed in
terms of a LS, where each term is computed explicitly using the BP solution.
Adapting this general result to the permanent, bulky yet straightforward algebra leads to
the following exact expression for the partition function $Z$ from \eq(\ref{GM}):
 \begin{eqnarray}
Z/Z_{BP}=z_{LS}; \qquad z_{LS}\equiv 1+\sum_{C \neq \emptyset} r_C, \nonumber  \\
   r_C \equiv \!\left(\prod_{i\in C}
     (1-q_i)\right)\!\!
   \left(\prod_{j\in C} (1-q^j)\right)\!\!\prod_{(i,j)\in C}
   \frac{\beta_i^j}{1-\beta_i^j}\,.
 \label{rC}
  \end{eqnarray}
The variables $\beta$ are in accordance with
\eqs(\ref{ds_cond},\ref{BP1})
 and $C$ stands for an arbitrary generalized loop, defined as
a subgraph of the complete bipartite graph with all its vertices having a degree larger than 1.
The $q_i$ (or $q^j$) in \eq(\ref{rC}) are the $C$-dependent degrees, i.e. $q_i=\sum_{j \mid (i,j)\in C} 1$ and $q^j=\sum_{i \mid (i,j)\in C}
1$. According to \eq(\ref{rC}), those loops with an even/odd number of vertices give positive/negative
contributions $r_C$.

\section{Loop Series as a Permanent}\label{sec:Per_BP_Per}

This section, explaining the main result of the paper, is split in two parts. In subsection
\ref{subsec:main} we give a simple derivation of a very compact representation for the LS
\eq(\ref{rC}) following directly from the BFE formulation. Subsection \ref{subsec:Ihara} contains
an alternative derivation of this main formula from LS using the concept of the Ihara-Bass graph
$\zeta$-function \cite{Idiscrete,Bass}.

We also find it appropriate here to make the following general remark. Even though discussion of the manuscript is limited to permanents,  counting perfect matchings over $K_{N,N}$, all the results reported in this section allows straightforward generalizations to weighted counting of perfect matchings over arbitrary (and not necessarily bipartite) graphs.

\subsection{Permanent representation for $Z/Z_{BP}$}
\label{subsec:main}
\begin{theorem}
\label{LS_new} For any non-integer solution of the BP equations,
\eqs(\ref{ds_cond},\ref{BP1}), the following is true:
\begin{equation}
 \perm (P)/Z_{BP}= \perm (\beta.*(1-\beta)) \prod_{(i,j)\in E}(1-\beta_i^j)^{-1}, \label{perm}
\end{equation}
where $A.*B$ is the element-by-element multiplication of  the $A$ and $B$ matrices.
\end{theorem}
\begin{proof}
From the definition of the BFE, ${\cal F}_{BP}=-T\ln Z_{BP}$, and 
\eqs(\ref{ds_cond},\ref{BP1}) one derives
\begin{equation*}
\fl
 Z_{BP}
 = \hspace{-2mm} \prod_{(i,j)\in E}\left[
(1-\beta_i^j) \Big( \frac{(p_i^j)^{1/T}}{\beta_i^j(1-\beta_i^j)} \Big)^{\beta_i^j}
\right]
  = \hspace{-2mm}
\prod_{(i,j)\in E}  \hspace{-1mm} (1-\beta_i^j) \prod_i \e^{- \mu_i}  \prod_j \e^{- \mu^j}.
\end{equation*}
On the other hand \eq(\ref{BP1}) results in, $\perm (P)= \perm (\beta.*(1-\beta))$
$ \prod_i \exp (- \mu_i)  \prod_j \exp (- \mu^j)$. Combining the two formulas we arrive at \eq(\ref{perm}).
\end{proof}

\begin{remark}{\rm
Note that if one considers expanding the permanent on the rhs of \eq(\ref{perm}) over the
elements of the matrix $\beta.*(1-\beta)$,  each element of the expansion will be positive, in
the contrast with the LS of \eq(\ref{rC}). Moreover,  the number of terms in the Perm-expansion
is significantly smaller than in the original LS.
}
\end{remark}

\subsection{From LS to the permanent representations for $Z/Z_{BP}$}
\label{subsec:Ihara}

Here we discuss the relation between the two complementary representations of $Z/Z_{BP}$, i.e.
between the LS expression (\ref{rC}) and the permanent formula (\ref{perm}).
We do this in two steps, stated in the two theorems presented consequently, one relating the LS to an average
of a determinant,  and another one expressing it via the permanent of $\beta.*(1-\beta)$.

\begin{theorem}
[LS as Average of Determinant]\label{thmA1}
Let $\vec{E}$ be the set of directed edges obtained by duplicating undirected edges $E$ of $K_{N,N}$.
Define the edge-adjacency matrix $\mathcal{M}$ of the complete bipartite graph $K_{N,N}$ according to
$\mathcal{M}_{\edij,\edkl}= \delta_{l,i}(1-\delta_{j,k})$.
Let $x=(x_{\edij})_{(\edij) \in \vec{E}}$
be the set of random variables that satisfies $\E{x_{\edij}}=0$, $\E{x_{\edij}x_{\edji}}=1$ and
$\E{x_{\edij}x_{\edkl}}=0 \quad (\{i,j\} \neq \{k,l\})$. (Here and below $\E{\cdots}_x$ stands for
the mathematical expectation over the random variables $x$.) Then the following relation holds:
$z_{LS} =\E{ \det [ I - i \mathcal{B} \mathcal{M}] }_x$, where $\mathcal{B}=\diag (\sqrt{\beta_i^j/(1-\beta_i^j)} x_{\edij})$.
\end{theorem}
\begin{proof}
For a general undirected graph $G$, the Ihara-Bass formula \cite{Idiscrete,Bass} states that
\begin{equation}
\zeta_{G}^{-1}(u)= \det[I-u \mathcal{M}]= (1-u)^{|E|-|{V}|} \det[I
+u^2(\mathcal{D}-I)-u\mathcal{A} ], \label{IB}
\end{equation}
where $\mathcal{A}$ is the adjacency matrix and $\mathcal{D}=\diag{(q_i;i\in V)}$ is the degree
matrix of $G$. If we take the limit $u \rightarrow \infty$, this formula implies
$\det{\mathcal{M}}=(-1)^{|E|} \prod_{i \in V}(1-q_i)$.
Expanding the determinant, one derives
\begin{equation}
 \det [ I - i \mathcal{B} \mathcal{M}]
=\sum_{\{ e_1,\ldots,e_n \} \subset \vec{E} } \det \mathcal{M}|_{ \{ e_1,\ldots,e_n \} }
(-i)^{k} \prod_{l=1}^n (\mathcal{B})_{e_l,e_l}. \label{DetExpansion}
\end{equation}
Evaluating the expectation of each summand in \eq(\ref{DetExpansion}), one observes that it is nonzero
only if $(\edij)  \in \{ e_1,\ldots,e_n \}$ implies $(\edji) \in \{ e_1,\ldots,e_n \}$,  thus
arriving at
\begin{equation*}
\fl
 \E{ \det [ I - i \mathcal{B} \mathcal{M}] }_x
= \sum_{C \subset E}   (-1)^{|C|}  \det \mathcal{M}|_C  \prod_{(i,j) \in C}
\frac{\beta_i^j}{1-\beta_i^j} = 1+\sum_{\emptyset \neq C \subset E} r_C.
\end{equation*}
\end{proof}

\begin{theorem}[From LS to Permanent]
For the doubly stochastic matrix of BP beliefs, $\beta$, and LS defined in \eq(\ref{rC}), one
derives
 \begin{equation*}
  z_{LS}= \perm (\beta.*(1-\beta)) \prod_{(i,j)\in E}(1-\beta_i^j)^{-1}.
 \end{equation*}
\end{theorem}
\begin{proof}
We use theorem \ref{thmA1}, choosing the random variables $x_i^j=x_{\edij}=x_{\edji}$
that take $\pm 1$ values with probability $1/2$.
We also utilize a multivariate version of
the Ihara-Bass formula from \cite{10WF} to derive the following expression for $z_{LS}$ proving the theorem
\begin{eqnarray*}
\fl
\det [ I - i \mathcal{B} \mathcal{M}]  =
\begin{small}
\det
\left[\begin{array}{cc}
 0& \sqrt{\beta.*(1-\beta) }.*x   \\
(\sqrt{\beta.*(1-\beta) }.*x)^{T}& 0 \\
\end{array}\right]\end{small}
\prod_{(i,j) \in E}(1-\beta_i^j)^{-1}, \\
\fl
 z_{LS}
= \E{\det(\sqrt{\beta.*(1-\beta) }.*x)^2}_x \prod_{(i,j)}(1-\beta_i^j)^{-1}
 = \perm (\beta.*(1-\beta)) \prod_{(i,j)}(1-\beta_i^j)^{-1}.\nonumber
\end{eqnarray*}
\end{proof}

\section{Invariance of the Gurvits-van der Waerden lower bound and new Lower Bounds for the Permanent}
\label{sec:low}

Van der Waerden \cite{26vdW} conjectured that the minimum of the permanent over the doubly
stochastic matrices is $N^N/N!$, and it is only attained when all entries of the matrix are
$1/N$. Though the conjecture appears to be simple, it remained open for over fifty years before
Falikman \cite{81Fal} and Egorychev \cite{81Ego} finally proved it.  Recently Gurvits
\cite{08Gur} found an alternative, surprisingly short and elegant proof, that also allowed a
number of unexpected extensions of the Van der Waerden conjecture. We call it the Gurvits-van
der Waerden theorem. (See e.g. \cite{09LS}.) A simplified form of this theorem is as follows.
\begin{theorem}[Gurvits-van der Waerden theorem \cite{08Gur,09LS}]
\label{T:Gurvits} For an arbitrary non-negative $N \times N$ matrix $A$,
\begin{eqnarray}
\fl
\perm (A) \geq \Mcap(p_A) \frac{N^N}{N!},\ \ \mbox{where}\ \
p_A(x) \equiv \prod_i \sum_j a_{i,j}x_j, \ \   \Mcap (p_A) \equiv \inf_{x \in \mathbb{R}^N_{>0}}
\frac{p_A(x)}{\prod_j x_j}.  \nonumber
\end{eqnarray}
\end{theorem}

We have found that the lower bound of the theorem \ref{T:Gurvits} has a ``good'' property with
respect to the BP transformation. As stated in theorem \ref{LS_new}, BP transforms the permanent to
another permanent. Therefore, applying theorem \ref{T:Gurvits} to both sides of \eq(\ref{perm}), one
naturally asks how do the two lower bounds compare? A somewhat surprising result is that the
Gurvits-van der Waerden theorem is invariant with respect to the BP transformation. Namely,
$\Mcap(p_P)= Z_{BP}* \Mcap(p_{\beta.*(1-\beta)}) \prod_{(i,j)\in E}(1-\beta_i^j)^{-1}$.
The lower bound for $\perm (\beta.*(1-\beta))$ based on the theorem \ref{T:Gurvits} is
\begin{corollary}
\label{Gurvits_bound}
\begin{equation*}
\perm (\beta.*(1-\beta))\geq \frac{N!}{N^N} \prod_{(i,j)\in E}(1-\beta_i^j)^{\beta_i^j}
\end{equation*}
\end{corollary}
\begin{proof}
This bound is the result of a direct application of the inequality, $\sum_j \beta_i^j(1-\beta_i^j)x_j
\geq \prod_j \left[(1-\beta_i^j)x_j \right]^{\beta_i^j}$, to theorem \ref{T:Gurvits}.
\end{proof}

We also obtain another lower bound which improves the bound of corollary \ref{Gurvits_bound}
at sufficiently low values of the temperature. See figure~\ref{fig}c for an
illustration.
\begin{theorem}
\label{second_low} For an arbitrary perfect matching $\Pi$ (permutation of $\{1,\ldots,N\}$),
\begin{equation*}
 \perm (\beta.*(1-\beta)) \geq
2 \prod_{i} \beta_i^{\Pi(i)}(1-\beta_i^{\Pi(i)})
\end{equation*}
\end{theorem}
\begin{proof}
Without loss of generality, we assume that $\Pi$ is the identity permutation. From the positivity
of entries and \eq(\ref{ds_cond}), we have $\perm (\beta.*(1-\beta)) \geq \prod_i \beta_i^i \perm (X)$,
where $X_{ij}=\delta_{i,j}+(1-2\delta_{i,j})\beta_i^j$.
Since $\beta$ is a stochastic matrix, $\det X=0$, and thus $\perm (X) \geq 2 \prod_i
(1-\beta_i^i)$.
\end{proof}
Note,  for the sake of completeness, that a comprehensive review of other bounds on permanents of specialized matrices (for example $0,1$ matrices) can be found in \cite{86LP}.

\section{New Upper Bound for Permanent}
\label{sec:up}

\begin{proposition}
\label{up}
\begin{equation*}
 \perm (\beta.*(1-\beta)) \leq  \prod_j({1-\sum_i (\beta_i^j)^2}).
\end{equation*}
\label{T:upper}
\end{proposition}
\begin{proof}
We use the Godzil-Gutman representation for permanents \cite{81GG}
\begin{equation}
 \perm(\beta.*(1-\beta)) = \E{ \det(\sqrt{\beta.*(1-\beta)}.*\sigma)^2 }_\sigma,
 \label{GG}
\end{equation}
where $\sigma_i^j=\pm 1$, with $i,j=1,\ldots,N$ are independent random variables taking values
$\pm 1$ equal probability. Each row of the matrix $\sqrt{\beta.*(1-\beta)}.*\sigma$ has the squared
Euclid norm ${\sum_i \beta_i^j (1-\beta_i^j)} = {1-\sum_i (\beta_i^j)^2}$. Therefore, the upper
bound is obtained from the Hadamard inequality, $|\det (a_1,\ldots,a_n) | \leq \norm{a_1}\cdots
\norm{a_n}$.
\end{proof}

\section{Path Forward}
\label{sec:path}

We consider this study to be the beginning of further research along the following lines:
(1) More detailed analysis of the BP solution. In particular,  study of $T_c$, e.g. concerning its dependence on the matrix size; analysis
of the BP solution dependence on temperature; and the construction of an iterative algorithm provably convergent to a nontrivial BP solution
for $T>T_c$. (2) Explanation of the BP invariance with respect to the Gurvits-van der Warden lower bound. (3) Development of a deterministic and/or randomized polynomial algorithm for estimating the permanent with provable guarantees based on the loop calculus expression. (4) Numerical tests of the lower and upper bounds for realistic large scale problems.

\ack
We are thankful to Leonid Gurvits for educating us, through his course of Lectures given at
CNLS/LANL, about existing approaches in the ``mathematics of the permanent''. YW acknowledges
support of the Students Visit Abroad Program of the Graduate University for Advanced Studies
which allowed him to spend two months at LANL and he is also grateful to CNLS at LANL for its
hospitality. Research at LANL was carried out under the auspices of the National Nuclear
Security Administration of the U.S. Department of Energy at Los Alamos National Laboratory under
Contract No. DE C52-06NA25396. MC also acknowledges partial support of NMC via the NSF collaborative
grant, CCF-0829945, on ``Harnessing Statistical Physics for Computing and Communications''.

\section*{References}
\bibliographystyle{iopart-num}
\bibliography{permanent,BP_review,zeta,MishaPapers}

\end{document}